\documentclass[submission,copyright,creativecommons]{eptcs}
\providecommand{\event}{AFL 2017} 
\usepackage{underscore}            
\usepackage{amssymb}
\usepackage{latexsym}
\usepackage{amsmath}
\usepackage{amsthm}
\usepackage[T1]{fontenc}
\usepackage{mathtools}
\usepackage{cite}
\usepackage{comment}
\usepackage{verbatim}
\usepackage{algorithm}
\usepackage[noend]{algorithmic}
\usepackage{algorithmic}
\usepackage{graphicx}
\usepackage{xcolor}
\usepackage{epsfig}

\DeclarePairedDelimiter\floor{\lfloor}{\rfloor}

\DeclarePairedDelimiter{\length}{\lvert}{\rvert}

\newtheorem{theorem}{Theorem}[section]

\newtheorem{lemma}[theorem]{Lemma}

\title{Constructing Words with High Distinct Square Densities}
\author{F. Blanchet-Sadri \qquad\qquad S. Osborne
\institute{Department of Computer Science\\
University of North Carolina\\
P.O. Box 26170\\
Greensboro, North Carolina 27402--6170, USA}
\email{blanchet@uncg.edu \qquad\qquad shosborn@uncg.edu}
}
\def\titlerunning{Constructing Words with High Distinct Square Densities}
\def\authorrunning{F. Blanchet-Sadri \& S. Osborne}
\begin{document}
\maketitle

\begin{abstract}
 Fraenkel and Simpson showed that the number of distinct squares in a word of length $n$ is bounded from above by $2n$, since at most two distinct squares have their rightmost, or last, occurrence begin at each position. Improvements by Ilie to $2n - \Theta(\log n)$ and by Deza et al. to $\lfloor 11n/6 \rfloor$ rely on the study of combinatorics of FS-double-squares, when the maximum number of two last occurrences of squares begin. In this paper, we first study how to maximize runs of FS-double-squares in the prefix of a word. We show that for a given positive integer $m$, the minimum length of a word beginning with $m$ FS-double-squares, whose lengths are equal, is $7m+3$. We construct such a word and analyze its distinct-square-sequence as well as its distinct-square-density. We then generalize our construction. We also construct words with high distinct-square-densities that approach $5/6$. 
\end{abstract}

\section{Introduction}

Computing repetitions in strings of letters from a finite alphabet is profoundly connected to numerous fields such as mathematics, computer science, and biology, where the data can be easily represented as words over some alphabet, and finds important practical uses in several research areas, notably in text compression, string searching and pattern matching \cite{CrHaLe,CrRy03}, cryptography, music, natural language processing \cite{Lot05}, and computational biology \cite{Gusbook,Pevbook}. Several pattern matching algorithms take advantage of the repetitions of the pattern to speed up the search of its occurrences in a text \cite{CrPe91,CrRy95} and algorithms for text compression are often based on the study of repetitions in strings \cite{LeZi}. We refer the reader to \cite{CrIlRy} for a survey on algorithms and combinatorics of repetitions in strings.  

There is a vast literature dealing with {\em squares}, which are repetitions of the form $xx$. This is due to their fundamental importance in algorithms and combinatorics on words. Different notions and techniques such as primitively or non-primitively-rooted squares \cite{DeFr,KRRW13}, positions starting a square \cite{HaKaNo}, frequencies of occurrences of squares \cite{KuOcRa,OcRa}, three-squares property \cite{FaPuSmTu,KoSm}, overlapping squares \cite{FFSS}, distinct squares \cite{DeFrTh,FrSi95,FrSi98,Ili05,Ili07,JoMaSe,MaSe}, double squares \cite{DeFrTh}, non-standard squares \cite{KoRaRyWa}, etc., have been studied and extended to partial words \cite{BSBoNiQuZh,BSJiMaQuZh,BSLaNiQuZh,BSMe1,BSMe2,BSMeSc2,BSNiQuZhConf,HaHaKa092}.

Various questions on squares have received a lot of attention. Among them is Fraenkel and Simpson's long-standing question ``How many distinct squares are there in a word of length $n$?'', where each square is counted only once \cite{FrSi95}. Fraenkel and Simpson \cite{FrSi98} showed in 1998 that the maximum number of distinct squares in such a word is asymptotically bounded from below by $n-o(n)$, and bounded from above by $2n$, since at each position of a word of length $n$ at most two distinct squares have their rightmost, or last, occurrence begin.  They conjectured that this maximum number is at most $n$. This work became the motivation for linear-time algorithms that find all repetitions in a string, encoded into maximal repetitions \cite{ApPr,Cro81,CrIlKuRaRyWa,GuSt,KoKu00,MaLo}. 

In 2005, Ilie \cite{Ili05} gave a simpler proof of the $2n$ upper bound and he \cite{Ili07} improved it to $2n-\Theta(\log n)$ in 2007. More recently, Deza et al. \cite{DeFrTh} improved the upper bound further to $\lfloor 11n/6 \rfloor$. Both Ilie's and Deza et al.'s improvements rely on the study of the combinatorics of {\em FS-double-square-positions}, i.e., positions at which two last occurrences of squares begin, which is the maximum number of occurrences possible. Let $s_i$ denote the number of distinct squares whose last occurrence in a word $w$ of length $n$ begins at position $i$, and let the distinct-square-sequence $s(w)$ be the word $s_1 s_2 \cdots s_n$. Then the result of Fraenkel and Simpson implies that $s_i \in \{0, 1, 2\}$. A position $i$ with $s_i=2$ is an FS-double-square-position. 

In this paper, we consider the problem of counting distinct squares in a word $w$ of length $n$. In particular, we study consecutive $2$'s in the sequence $s(w)$, called {\em runs of 2's}. We also construct words that have a high distinct-square-density, that is, the ratio of the number of distinct squares to length is high.   

The contents of our paper are as follows. 
In Section~2, we review some basic definitions and notations that we use throughout the paper. We also discuss some preliminary results on double-squares.  
In Section~3, we study runs of double-square-positions and we focus on maximizing runs of FS-double-squares. We first recall two results of Ilie~\cite{Ili07}; one gives a relation between the lengths of squares having their last occurrence at positions neighboring an FS-double-square-position and the other one considers the case when the lengths of squares in a run of 2's are equal. It follows from the latter that for a given $m$, the minimum length of a word beginning with $m$ FS-double-squares, whose lengths are preserved, is $7m+3$. We show its existence by constructing one, i.e., we construct a word $w_m$ of length $7m+3$ beginning with $m$ FS-double-squares, whose lengths are preserved, and analyze the distinct-square-sequence $s(w_m)$ as well as the distinct-square-density of $w_m$. We then generalize our construction.
In Section~4,  we construct words in which the distinct-square-density approaches $5/6$. These words do not have many FS-double-squares, and those they do have are not at the beginning.  The majority of distinct squares in these words are the only distinct squares at a particular position. All our constructions in Sections~3 and 4 are such that  each run of 2's in the corresponding distinct-square-sequence is followed by a run of at least twice as many 0's.  We refer to such a run of 2's as {\em selfish 2's}.
In Section~5, we discuss ways to break the selfish rule, e.g., omit or alter the last letter of the word $w_m$.
Finally in Section~6, we conclude with some remarks and suggestions for future work.

\section{Preliminaries}

We refer the reader to the book \cite{Lot97} for some basic concepts in combinatorics on words. We also adopt some of the terminology of \cite{DeFrTh,FrSi98,Ili07} on squares. For integers $i, j$ such that $i \leq j$, the notation $[i..j]$ denotes the discrete interval consisting of the integers $\{i, i+1, \ldots, j\}$.

Let $A$ be an alphabet with size denoted by $|A|$; we assume throughout the paper that $|A| \geq 2$. A {\em word} $w$ over $A$ is a sequence $a_1 \cdots a_n$, where $a_i$ is the letter in $A$ that occurs at position~$i$ of $w$; we also let $w[i]$ denote the letter at position~$i$. The integer $n$ is the {\em length} of $w$, denoted by $\length{w}$. The {\em empty word}, denoted by $\varepsilon$, is the word of length zero. It acts as the identity under the concatenation of words, so the set of all words over $A$, denoted by $A^*$, becomes a monoid. 

If $w = xy$, then $x$ is a {\em prefix} of $w$, denoted by $x \leq w$; when $x \not = w$, we say that $x$ is a {\em proper prefix} of $w$, denoted by $x < w$. If $w = x y z$, then $y$ is a {\em factor} of $w$ and $z$ is a {\em suffix} of $w$; here $y$ is an {\em interior factor} of $w$ if $x \not = \varepsilon$ and $z \not = \varepsilon$.  For $1 \leq i \leq j \leq |w|$, the notation $w[i..j]$ refers to the factor $w[i] w[i+1] \cdots w[j]$.

A word $w$ is {\em primitive} if it cannot be written as a non-trivial power $v^e$, i.e., the concatenation of $e$ copies of a word $v$ where $e$ is an integer greater than 1. It is well-known that this is equivalent to saying that $w$ is not an interior factor of $ww$.

A {\em square} in a word consists of a factor of the form $x^2=xx$, where $x \not = \varepsilon$.  A {\em double-square} is a pair $(u,U)$ such that $u^2$ and $U^2$ are two squares that begin at the same position with $\length{u}<\length{U}$. An {\em FS-double-square} is a double-square $(u,U)$ positioned such that both $u^2$ and $U^2$ are last occurrences; the FS notation refers to work of Fraenkel and Simpson on double-squares.  Note that all FS-double-squares $(u,U)$ are such that $\length{u}<\length{U}<2\length{u}$. A {\em double-square-position} is a position at which a double-square begins, and likewise for FS-double-squares. The structure of an FS-double-square follows Lemma~\ref{Lem Deza 6} since the third condition is always true, by definition.

\begin{lemma}\cite{DeFrTh}
	\label{Lem Deza 6}
	Let $(u, U)$ be a double-square such that $|u| < |U| < 2|u|$.  If one of the conditions $1.$ $u$ is primitive, $2.$ $U$ is primitive, or $3.$ $u^2$ has no further occurrence in $U^2$ is satisfied, then there is a unique primitive word $v_1$, a unique non-empty proper prefix $v_2$ of $v_1$, and unique integers $e_1$ and $e_2$ satisfying $1 \leq e_2 \leq e_1$ such that $u=v_1^{e_1}v_2$ and $U=v_1^{e_1}v_2v_1^{e_2}$. Moreover, $U$ is primitive.
\end{lemma}

Note that by Lemma~\ref{Lem Deza 6}, $U^2$ has the form $v_1^{e_1}v_2v_1^{e_2}v_1^{e_1}v_2v_1^{e_2}$.  Any FS-double-square $(u,U)$ such that $|u| < |U| < 2|u|$ can therefore be fully defined by giving values to the terms $v_1, v_2, e_1$, and $e_2$.  When defining a word in such a manner, we will list the terms in the order just given; e.g., $(ba, b, 2, 1)=bababbabababba$. 

The converse of Lemma~\ref{Lem Deza 6} is nearly true as well, as the next lemma shows.

\begin{lemma}
	\label{Lem Deza converse}
	For any word $w$ that can be written as $w=(v_1^{e_1}v_2v_1^{e_2})^2$, where $v_1$ is primitive, $v_2$ is a proper non-empty prefix of $v_1$, and $e_1$ and $e_2$ are integers such that $1 \leq e_2 \leq e_1$, the following hold: $1.$ $U=v_1^{e_1}v_2v_1^{e_2}$ is primitive, $2.$ 
 $u^2=(v_1^{e_1}v_2)^2$ occurs at position $1$ of $w$ and has no further occurrence in $w$, and $3.$ $(u,U)$ is an FS-double-square such that $|u| < |U| < 2|u|$.
		\end{lemma}
	\begin{proof}
			Deza et al. \cite[Lemma~6]{DeFrTh} prove that $U$ must be primitive. Since $v_2$ is a prefix of $v_1$, it is clear that $u^2=(v_1^{e_1}v_2)$ occurs at position 1 of $w$.  If $u^2$ occurs anywhere else in $w$, then $v_1$ is a power of $v_2$, in which case $v_1$ is not primitive, $v_1=v_2$, or $v_2$ is the empty word, a contradiction. With $w=U^2$, and $u^2$ occurring exactly once in $w$, at its first position, $w$ meets the definition of an FS-double-square.
		\end{proof}
	
				Referring to Lemma~\ref{Lem Deza converse}, note that $u$ may or may not be primitive. An example of $u$ being primitive is achieved by letting $v_1=a^{m-1}ba$, $v_2=a^{m-1}b$, and $e_2=e_1=1$, where $m$ is an integer greater than or equal to 1. 	An example of $u$ not being primitive is shown by letting $v_1=abaaaba$, $v_2=a$, and $e_2=e_1=1$, giving the non-primitive word $u=abaaabaa$ and $w=abaaabaaabaaabaabaaabaaabaaaba$ has the distinct-square-sequence $200111011101111000011110001000$.

If a word $w$ has length $n$, we let $d(w)$ denote the {\em distinct-square-density} of $w$ equal to $d(w)=s/n$, where $s$ is the number of distinct squares in $w$. We also let $s(w)$ denote the sequence $s_1\cdots s_{n}$ such that each $s_i$ is the number of distinct squares whose last occurrence in $w$ begins at position $i$. As mentioned earlier, $s_i \in \{0, 1, 2\}$. If $s_i = 2$, then $i$ is an FS-double-square-position. In Figure~\ref{exsquares}, the sequence $s(w)=s_1\cdots s_{17}=22000011100110010$ is such that each $s_i$ is the number of distinct squares whose last occurrence in the word $w$ begins at position $i$ (there are 10 distinct squares in that word of length 17). There are two FS-double-square-positions: position~1 with FS-double-square $(abaab,abaababa)$ and position~2 with $(baaba,baababaa)$. There are also six positions with one square whose last occurrence begins: position~7 with $(baa)^2$, position~8 with $(aab)^2$, position~9 with $(aba)^2$, position~12 with $(ab)^2$, position~13 with $(ba)^2$, and position~16 with $a^2$, and all other positions are such that no last occurrence of a square begins.  

\begin{figure}
\begin{center}
$\begin{array}{c|ccccccccccccccccccc}
i & 1 & 2 & 3 & 4 & 5 & 6 & 7 & 8 & 9 & 10 & 11 & 12 & 13 & 14 & 15 & 16 & 17 \\
\hline
w[i] & a & b & a & a & b & a & b & a & a & b & a & a & b & a & b & a & a \\
\hline
s_i & 2 & 2 & 0 & 0 & 0 & 0 & 1 & 1 & 1 & 0 & 0 & 1 & 1 & 0 & 0 & 1 & 0\\ 
\end{array}$
\end{center}
\caption{Word $w$ of length 17 with distinct-square-sequence $s(w)=s_1 \cdots s_{17}$, where $s_i$ is the number of distinct squares whose last occurrence in $w$ begins at position~$i$. Up to a renaming of letters, this is the shortest word that begins with two consecutive double-square-positions.}\label{exsquares}
\end{figure}
	
\section{Runs of double-square-positions}

Ilie \cite{Ili07} gave the following relation between the lengths of squares having their last occurrence at positions neighboring an FS-double-square-position. 

\begin{lemma} [\cite{Ili07}]
	If $(u,U)$ is an FS-double-square beginning at position $i$ and $w^2$ is a square with last occurrence beginning at position $i+1$, then either $\length{w} \in \{\length{u}, \length{U}\}$ or $\length{w} \geq 2\length{u}$. 
\end{lemma}

Ilie \cite{Ili07} also considered the case when the lengths of squares in a run of 2's are equal. His lemma was originally stated for $m\geq 2$, but it also holds for $m=1$.
	
\begin{lemma} [\cite{Ili07}]
	\label{LemIlie3}
	Let $m \geq 1$ be such that $i$ is an FS-double-square-position for all $i \in [1..m]$, and let $(u_i,U_i)$ be the FS-double-square at $i$.  If $|u_1|=\cdots=|u_m|$ and $|U_1|=\cdots=|U_m|$, then for all $i \in [1..m]$ the following hold:
	\begin{enumerate}
	\item $\length{U_i}+m \leq 2\length{u_i}$, 
	\item $\length{U_i} \geq \length{u_i}+m+1$,  
	\item $\length{U_i} \geq 3m+2$ and $\length{u_i} \geq 2m+1$.
	\end{enumerate}
\end{lemma}

From Lemma~\ref{LemIlie3}, it follows that for a given $m$, the minimum length of a word beginning with $m$ FS-double-squares, whose lengths are preserved, is $7m+3$. The existence of such a word can be easily verified by constructing one. Theorem~\ref{Thm unique} constructs a word $w_m$ of length $7m+3$ with a prefix of $m$ FS-double-square-positions; the number $m$ of initial FS-double-square-positions is maximum for a word of that length. We do not claim $w_m$ has the largest possible number of distinct squares given its length (words with higher 
distinct-square-densities will be constructed in the next section).  

\begin{theorem}
	\label{Thm unique}
	Let $m \geq 1$. Then there exists a word $w_m$ of length $7m+3$ such that $i$ is an FS-double-square-position with double-square $(u_i, U_i)$ for all $i \in [1..m]$, where $\length{u_1} = \cdots = \length{u_m} = 2m+1$ and $\length{U_1}= \cdots = \length{U_m}=3m+2$. This word is unique up to a renaming of letters, i.e., $$w_m=(a^{m-1}baa^{m-1}ba^{m-1}ba)^2a^{m-1}$$ where $a$ and $b$ are distinct letters of the alphabet.  
\end{theorem}
\begin{proof}
		We construct such a word $w_m$. Note that $w_m$ must contain at least two distinct letters, otherwise the squares would appear later. Let $u_1=a^{m-1}baa^{m-1}b$ and $U_1=u_1a^{m-1}ba$, and let $$w_m=U_1^2a^{m-1}=(a^{m-1}baa^{m-1}ba^{m-1}ba)^2a^{m-1}.$$
		
		By \cite[Definition~10]{DeFrTh}, recall that a factor $u=x[{i'}..{j'}]$ of a word $x$ can be cyclically shifted right by 1 position if $x[{i'}]=x[{j'}+1]$.  The factor $u$ can be cyclically shifted right by $k$ positions if $u$ can be cyclically shifted right by 1 position and the factor $x[{i'}+1..{j'}+1]$ can be cyclically shifted right by $k-1$ positions.  This similarly holds for left cyclic shifts.  A trivial cyclic shift is a shift by 0 positions. 
		
		It is easy to see that the last and only occurrences of both $u_1^2$ and $U_1^2$ in $w_m$ are at position~1.  Furthermore, both $u_1^2$ and $U_1^2$ can be cyclically shifted right $m-1$ times, such that $u_i=a^{m-i}baa^{m-1}ba^{i-1}$ and $U_i=a^{m-i}baa^{m-1}ba^{m-1}ba^i$.  With this shift, both $u_i^2$ and $U_i^2$ have last occurrences at position~$i$ of $w_m$, for all $i \in [1..m]$.  Thus, $w_m$ begins with $m$ FS-double-square-positions. 
			
		We claim that $w_m$ is unique up to a renaming of letters. Let $w=a_1 \cdots a_{7m+3}$ be a word that satisfies the requirements.	We have $u_1=a_1 \cdots a_{2m+1}$ and $U_1=u_1a_{2m+2}\cdots a_{3m}a_{3m+1}a_{3m+2}=u_1a_1\cdots a_{m-1}a_{m}a_{m+1}$, and $w=U_1^2 a_{6m+5}\cdots a_{7m+3}$.
		
		For all $i \in [1..m]$, we have the square $u_i^2$ of length $4m+2$ and the square $U_i^2$ of length $6m+4$ both beginning at position~$i$, so $a_i \cdots a_{i+2m}= a_{i+2m+1} \cdots a_{i+4m+1}$ and $a_i \cdots a_{i+3m+1}= a_{i+3m+2} \cdots a_{i+6m+3}$. This implies that for all such $i$, we have $a_i=a_{i+2m+1}=a_{i+3m+2}=a_{i+5m+3}=a_{i+m+1}=a_{i+4m+3}$. 
		
		We next show that the first $m-1$ positions of $w$ each have the same letter, i.e., $a_1=\cdots = a_{m-1}$. To do this, we show that for all $i \in [2..m-1]$, we have $a_i=a_{i-1}$. So let $i \in [2..m-1]$. Recalling that $\length{u_i}=2m+1$ and $\length{U_i}=3m+2$, an FS-double-square $(u_i, U_i)$ beginning at position $i$ implies that $a_i=a_{i+2m+1}$ and an FS-double-square $(u_{i+1}, U_{i+1})$ at position $i+1$ implies that the following letters are equal: 
		\begin{center}
		$\begin{array}{rcl}
			a_i=a_{i+2m+1}=w[(i+1)+2m]&=&w[((i+1)+2m)+(2m+1)]\\
			&=&w[(i+1)+4m+1]\\
			&=&w[((i+1)+4m+1)-(3m+2)]\\
			&=&w[(i+1)+m-1]\\
			&=&w[((i+1)+m-1)+(2m+1)]\\
			&=&w[(i+1)+3m]\\
			&=&w[((i+1)+3m)-(3m+2)]\\
			&=&w[i-1]=a_{i-1}.
			\end{array}$
		\end{center}

	It follows that all positions in $w$ other than $m, 2m+1, 3m+1, 4m+2, 5m+3$, and $6m+3$ must have the same letter.  It can be easily verified that the remaining six positions must all have the same letter, and that they may not have the same letter as position~1. Our claim follows. 
	\end{proof}
	
	By Lemma~\ref{Lem Deza 6}, the first FS-double-square $(u_1,U_1)$ of $w_m$ of Theorem~\ref{Thm unique} satisfies $u_1=a^{m-1}baa^{m-1}b=v_1^{e_1}v_2$ and $U_1=a^{m-1}baa^{m-1}ba^{m-1}ba=v_1^{e_1}v_2v_1^{e_2}$, where $v_1=a^{m-1}ba$, $v_2=a^{m-1}b$, and $e_1=e_2=1$. As discussed in Section~2, we can write $w_m=(a^{m-1}ba, a^{m-1}b, 1, 1)a^{m-1}$.  
	
	The word in Figure~\ref{exsquares} is the $(m=2)$-case of Theorem~\ref{Thm unique} whose distinct-square-sequence can be generalized as follows.
	
	\begin{theorem}
	\label{Thm run}
	For $w_m$ as in Theorem~\ref{Thm unique}, $$s(w_m)=2^m0^{2m}1^{m+1}001^m0^{2\lfloor \frac{m+1}{2}\rfloor} (10)^{\lfloor \frac{m}{2} \rfloor}.$$
	\end{theorem}
	\begin{proof}
	As noticed earlier, $w_m$ begins with $m$ FS-double-square-positions, so $s(w_m)[1..m]=2^m$. 

Let us look at the $aa^{m-1}$-suffix of $w_m$. We have one position where the last occurrence of the square $(a^i)^2$, where $1 \leq i \leq \lfloor \frac{m}{2} \rfloor$, begins, so the corresponding position in the suffix of $s(w_m)$ is a 1. Each of the other $\lfloor \frac{m}{2} \rfloor$ positions in the $aa^{m-1}$-suffix of $w_m$ corresponds to a 0, yielding a suffix of $(10)^{\lfloor \frac{m}{2} \rfloor}$ for $s(w_m)$. Looking at the $a^{m-1}baa^{m-1}$-suffix of $w_m$, the remaining positions must correspond to 0 since any square cannot contain only one $b$, and the other squares only contain $a$'s but appear later. So $s(w_m)[5m+4..7m+3]=0^{2\lfloor \frac{m+1}{2}\rfloor} (10)^{\lfloor \frac{m}{2} \rfloor}$.

Next, let us look at $w_m[4m+4..7m+3]=a^{m-1}ba^{m-1}baa^{m-1}$. If a square contains only $a$'s, its last occurrence has already been discussed. Otherwise, it contains $b$'s. Each of the positions $4m+4, \ldots, 5m+3$ give the last occurrence of the square $(a^{m-1}b)^2$ or one of its rotations. So $s(w_m)[4m+4..5m+3]=1^m$. It is easy to see that $s(w_m)[4m+2..4m+3]=00$. Also, $s(w_m)[3m+1..4m+1]=1^{m+1}$ due to the last occurrence of the square $(baa^{m-1})^2$ and its rotations.

Finally, let us look at $w_m[m+1..7m+3]$. Either the squares have two $b$'s or four $b$'s or they have only $a$'s.
The case of four $b$'s is impossible and the case of only $a$'s have their last occurrence later. In the case of two $b$'s, the square $(a^{m-1}b)^2$ or its rotations at positions $m+2, \ldots, 2m+1$ appear later. Similarly,  the square $(a^{m-1}ba)^2$ and its rotations at positions $2m+2, \ldots, 3m$ appear later. It is easy to check that no square has its last occurrence beginning at position~$m+1$. So $s(w_m)[m+1..3m]=0^{2m}$. 
\end{proof}			

For $w_m$ as in Theorem~\ref{Thm unique}, $d(w_m)=\frac{4.5m+1}{7m+3}$ for even values of $m$ and $d(w_m)=\frac{4.5m+.5}{7m+3}$ for odd values of $m$.  In both cases, for large values of $m$ the distinct-square-density of $w_m$ is approximately .643.  
				
From Theorem~\ref{Thm unique}, it follows that the probability of a word of length $7m+3$ beginning with a run of $m$ FS-double-square-positions, where the lengths of the squares are preserved, is equal to $\frac{|A|^2-|A|}{|A|^{7m+3}}$.

	While we prove the converse of Lemma~\ref{LemIlie3} in a specific case, the converse is not true in general. That is, in the proof of Theorem~\ref{Thm unique} we construct, for every $m\geq 1$, a word starting with $m$ FS-double-squares $(u_i,U_i)$, whose lengths are preserved, such that $1'.$ $\length{U_i}+m = 2\length{u_i}$, $2'.$ $\length{U_i} = \length{u_i}+m+1$,  and $3'.$ $\length{U_i} = 3m+2$ and $\length{u_i} = 2m+1$, but the construction may be impossible if we replace the equalities $1'-2'-3'$ with Lemma~\ref{LemIlie3}'s inequalities $1-2-3$.  For example, given $m=1$, $\length{u_1}=6$, and $\length{U_1}=8$, all three conditions of Lemma~\ref{LemIlie3} are satisfied, but it is impossible to construct a word starting with one FS-double-square fulfilling those criteria.  The same is true for $m=1$, $\length{u_1}=6$, and $\length{U_1}=9$ or $m=2$, $\length{u_1}=\length{u_2}=6$, and $\length{U_1}=\length{U_2}=9$. However, the following theorem holds.

\begin{theorem}
	\label{Thm larger alphabet}
	Let $m \geq 1$ and $\ell \geq m$. Then there exist at least $\length{A}(\length{A}-1)^{\ell-m+1}$ and fewer than $\length{A}^{\ell-m+2}$ words of length $6\ell+m+3$ over an alphabet $A$, with $|A| \geq 2$, such that $i$ is an FS-double-square-position with double-square $(u_i, U_i)$ for all $i \in [1..m]$, where $\length{u_1} = \cdots = \length{u_m} = 2\ell+1$ and $\length{U_1}= \cdots = \length{U_m}=3\ell+2$.  
	\end{theorem}
	\begin{proof} Such a word $w$ can be constructed as follows, where ${\ell'}=\ell-m+2$ and $\{\alpha_2, \ldots, \alpha_{\ell'}\}$ denotes a set of ${\ell'}-1$ letters distinct from $\alpha_1=a$.  Set $w=a_1 \cdots a_{6\ell+m+3}$ where each $a_i$ is a letter of the alphabet. Since $u_1^2=(a_1 \cdots a_{2\ell+1})^2$ and $U_1^2=(a_1 \cdots a_{3\ell+2})^2$ are squares starting at position~1, we deduce that for all $i \in [1..\ell+1]$, we have $a_i= a_{i+2\ell+1}$ and for all $i \in [1..\ell]$, we have $a_i=a_{i+\ell+1}$. Since all the $m$ first positions of $w$ are FS-double-square-positions, we also have $a_i = a_{i+\ell}$  for all $i \in [1..m-1]$. 
	
Thus for all $i \in [1..m-2]$, $a_i=a_{i+\ell+1}=a_{i+1}$, so set $a_1= \cdots=a_{m-1}=a$. Also set $a_m \cdots a_{\ell}=\alpha_2 \cdots \alpha_{\ell'}$. We obtain $w=U_1^2a^{m-1}$ where
	\begin{center}
			$\begin{array}{rcl}
			u_1&=&a^{m-1}\alpha_2\cdots\alpha_{\ell'}aa^{m-1}\alpha_2\cdots\alpha_{\ell'},\\
			U_1&=&a^{m-1}\alpha_2\cdots\alpha_{\ell'}aa^{m-1}\alpha_2\cdots\alpha_{\ell'}a^{m-1}\alpha_2\cdots\alpha_{\ell'}a.	
		\end{array}$
		\end{center}
		
		As with Theorem~\ref{Thm unique}, it is easy to see that the last and only occurrences of both $u_1^2$ and $U_1^2$ in $w$ are at position~1 of $w$.  Furthermore, both $u_1^2$ and $U_1^2$ can be cyclically shifted right $m-1$ times, such that 
		\begin{center}
		$\begin{array}{rcl}
		u_i&=&a^{m-i}\alpha_2\cdots\alpha_{\ell'}aa^{m-1}\alpha_2\cdots\alpha_{\ell'}a^{i-1},\\
		U_i&=&a^{m-i}\alpha_2\cdots\alpha_{\ell'}aa^{m-1}\alpha_2\cdots\alpha_{\ell'}a^{m-1}\alpha_2\cdots\alpha_{\ell'}a^i.
			\end{array}$
			\end{center}
			With this shift, both $u_i^2$ and $U_i^2$ have last occurrences at position~$i$ of $w$, for all $i \in [1..m]$.  Thus, $w$ begins with $m$ FS-double-square-positions. 
			
			The minimum number of possible words $\length{A}(\length{A}-1)^{\ell-m+1}$ is calculated by allowing $\alpha_1$ to be any letter of alphabet $A$ and allowing each letter $\alpha_i$, with $i \in [2..\ell-m+2]$ to be any letter of $A$ distinct from $\alpha_1$. The exclusive maximum number of words is calculated by allowing each letter $\alpha_i$ to be any letter of $A$ including $\alpha_1$.
				
				Note that in the case of $\ell=m$, this word $w$ is identical to $w_m$ given in Theorem~\ref{Thm unique}. Note also that the proof holds even if the letters $\alpha_2, \ldots, \alpha_{\ell'}$ are not distinct. To see this, letting $\alpha_2, \ldots, \alpha_{\ell'}$ be some letters distinct from $a=\alpha_1$, we have $w=U_1^2a^{m-1}$ where
		\begin{center}
		$\begin{array}{rcl}
		u_1&=&a^{m-1}b^{\ell-m+1}aa^{m-1}b^{\ell-m+1},\\
		U_1&=&a^{m-1}b^{\ell-m+1}aa^{m-1}b^{\ell-m+1}a^{m-1}b^{\ell-m+1}a.		
		\end{array}$
		\end{center}
	\end{proof}

The next theorem, which constructs FS-double-squares, extends Theorem~\ref{Thm larger alphabet}.

\begin{theorem}
	\label{Thm Z words}
	Let $m \geq 1$, $Z$ be a non-empty word such that $a^{m-1}Za$ is primitive, $w=(v_1^{e_1}v_2v_1^{e_2})^2a^{m-1}$ where $v_1=a^{m-1}Za$, $v_2=a^{m-1}Z$, $e_1$ and $e_2$ are integers such that $1 \leq e_2 \leq e_1$.  Then $w$ begins with $m$ FS-double-squares $(u_i, U_i)$, $i \in [1..m]$, where 
	\begin{center}
	$u_i=(a^{m-i}Za^i)^{e_1}a^{m-i}Za^{i-1}$ and $U_i=(a^{m-i}Za^i)^{e_1}a^{m-i}Za^{i-1}(a^{m-i}Za^i)^{e_2}$.
	\end{center}
	\end{theorem}
	\begin{proof}
	Let $i \in [1..m]$.  Set $v_{1,i}=a^{m-i}Za^i$ and $v_{2,i}=a^{m-i}Za^{i-1}$. Then $v_{1,i}$ is primitive being a cyclic shift of $v_1$, and  $v_{2,i}$ is a proper non-empty prefix of $v_{1,i}$. By Lemma~\ref{Lem Deza converse}, the word $(v_{1,i}^{e_1}v_{2,i}v_{1,i}^{e_2})^2$, where $e_1$ and $e_2$ are integers such that $1 \leq e_2 \leq e_1$, is an FS-double-square.  
		
		We claim that adding any number of $a$'s to the end of $(v_1^{e_1}v_2v_1^{e_2})^2$ will not destroy the initial FS-double-square, since a $Z$ would be required to create an additional $(v_1^{e_1}v_2)^2$ or $(v_1^{e_1}v_2v_1^{e_2})^2$ to the right of the initial FS-double-square. It follows that $w$ begins with an FS-double-square $(u_1,U_1)$.  Furthermore, both $u_1^2$ and $U_1^2$ can be cyclically shifted right $m$ times, to create $m$ double-squares $(u_i,U_i)$ such that $u_i=v_{1, i}^{e_1}v_{2, i}$ and $U_i=v_{1, i}^{e_1}v_{2, i}v_{1, i}^{e_2}$.
		
		Thus, FS-double-squares are found at each of the first $m$ positions of $w$. 
	\end{proof}
	
	In Theorem~\ref{Thm Z words}, note that in the case where $Z$ begins with $a$, adding an additional $a$ to the end of $w$ will produce a word that begins with an additional FS-double-square.
	
Note that our constructions have focused on a run of 2's in the prefix in which the lengths of the double-squares remain the same. The following word 
\begin{center}
$\begin{array}{cccccccccccccccccccc}
b&a&b&b&a&b&a&b&b&a&a&a&b&b&a&b&a&b&b&a\\
2&2&0&0&0&0&0&0&0&0&1&1&1&0&0&0&1&0&0&1\\
&&&&&&&&&&&&&&&&&&&\\
a&b&b&a&b&a&b&b&a&a&a&b&b&a&b&a&b&b&a&\\
0&0&0&0&0&0&0&0&0&1&0&0&1&1&0&0&1&0&0&
\end{array}$
\end{center}
has a distinct-square-sequence with two consecutive 2's in the prefix that refer to double-squares of different lengths, $(bab,babba)$ and $(abbababbaa,abbababbaaabbababba)$.

\section{Constructions with higher distinct-square-densities}

	For any pair of integers $(i, j)$ with $i<j$, let $Y_{i,j}=X_i X_{i+1} \cdots X_{j-1}X_jaa^{j-1}$, where $$X_k=a^{k-1}baa^{k-1}ba^{k-1}baa^{k-1}baa^{k-1}ba^{k-1}b$$ for each $k \in [i..j]$. We will show that the distinct-square-density of the word $Y_{i,j}$ approaches $5/6$ as $j$ approaches infinity.

	\begin{lemma}
	\label{Lem Xk}
		Each factor $X_k$ with $k<j$ has, in $Y_{i,j}$, the distinct-square-sequence $1^{2k+1}0^{k+1}1^k01^{2k}$, giving $5k+1$ distinct squares per $X_k$. No $X_k$, where $k<j$, has more distinct squares than those listed in Table~\ref{Tab Xk}.
	\end{lemma}
		
	\begin{proof}
		Consider first the case where $j=i+1$.  In this case, $Y_{i,j}=X_i X_{i+1}aa^i$.   Table~\ref{Tab Xk} lists the distinct squares of $Y_{i,j}$ that begin in $X_i$. The squares are listed as sets of consecutive distinct squares of the same length.  For each set, only the first distinct square is listed; the remaining distinct squares are cyclic shifts of the first. Both the existence and distinctness of the listed squares may be easily verified. The sequence $1^{2k+1}0^{k+1}1^k01^{2k}$ and the minimum total of $5k+1$ in $X_i$ both follow from Table~\ref{Tab Xk}.
		
		By Theorem~\ref{Thm run}, $s(X_jaa^{j-1})=s(w_j)=2^j0^{2j}1^{j+1}001^{j}0^{2\floor{\frac{j+1}{2}}(10)^{\floor{\frac{j}{2}}}}$.
		
		\begin{table}
			\begin{center}
				
				\begin{tabular}{l|l|l|l|l}
					first in set & last in set & count & root of first in set & root length\\ 
					\hline
					\hline
					$1$ 		& $2k+1$ & $2k+1$ & $a^{k-1}baa^{k-1}ba^{k-1}ba$ & $3k+2$ \\
					$2k+2$	& $3k+2$ & $k+1$ & NONE & 0\\
					$3k+3$	& $4k+1$ & $k-1$ & $a^{k-1}ba$ & $k+1$\\
					$4k+2$	& $4k+2$ & $1$	& $baa^{k-1}ba^{k-1}$ & $2k+1$\\
					$4k+3$	& $4k+3$ & $1$	& NONE & 0\\
					$4k+4$	& $5k+3$ & $k$ & $a^{k-1}b$ & $k$\\
					$5k+4$	& $6k+3$ & $k$ & $a^{k-1}ba^{k}baa$ & $2k+3$\\
					\end{tabular}
				\caption{Distinct squares of $Y_{i,j}$ that begin in the factor $X_k=a^{k-1}baa^{k-1}ba^{k-1}baa^{k-1}baa^{k-1}ba^{k-1}b$, where $k=j-1$: each row lists a set of consecutive distinct squares of a given length, with the first square of the set beginning at the first position and the last beginning at the last position.  ``Count'' gives the number of distinct squares in the set.  When the root of the square is given as NONE, no distinct square exists at those positions. Here $s(Y_{i,j})=s(\cdots X_kX_jaa^{j-1})=\cdots 1^{2k+1}0^{k+1}1^k01^{2k}s(X_jaa^{j-1})$.}\label{Tab Xk}
			\end{center}
			\end{table}
		
		Now consider the case where $Y_{i,j}$ contains more than two $X_k$ factors. By definition, every $X_k$ has the same structure, and for all $k<j$, every factor $X_k$ is followed by $X_{k+1}aa^{k}$, by the definition of $Y_{i, j}$. Note that when $k+1 < j$, the factor $aa^k$ is a prefix of $X_{k+2}$. Since every $X_k$ has the same structure and is followed by $X_{k+1}aa^{k}$, the squares listed in Table \ref{Tab Xk} will exist in all $X_k$'s.
		
		To see that the squares in each $X_k$ are distinct regardless of how many $X_k$ factors are in $Y_{i,j}$, observe that every distinct square in $X_k$, as listed in Table \ref{Tab Xk}, includes at least one of the factors $ba^{k-1}b$ or $baa^{k-1}b=ba^kb$.
		
		Now consider $X_{k+2}=a^{k+1}baa^{k+1}ba^{k+1}baa^{k+1}baa^{k+1}ba^{k+1}b$.  Neither $ba^{k-1}b$ nor $ba^kb$ appears in $X_{k+2}$, nor will they appear in subsequent $X$'s.  Thus the distinct squares of $X_k$ listed in Table~\ref{Tab Xk} remain distinct no matter how many additional $X$ factors exist in $Y_{i,j}$. Table~\ref{Tab Xk} already accounts for $X_{k+1}$.
		
		Table~\ref{Tab Xk} holds for $Y_{k,k+1}=X_{k}X_{k+1}aa^{k}$. It can be verified that no distinct squares of $X_kX_{k+1}aa^k$ that begin in $X_k$ exist beyond those listed in Table~\ref{Tab Xk}.  If we want additional distinct squares in a $Y_{i,j}$ word, we must add an additional $X_k$ factor (increase the value of $j$), giving us $Y_{i,j}=X_kX_{k+1}X_{k+2}aa^{k+1}$.
		
		We are looking for additional distinct squares that begin in $X_k$.  We know from the case $j=i+1$, shown in Table~\ref{Tab Xk}, that all distinct squares beginning in $X_k$ and ending in $X_k, X_{k+1}$, or the prefix of length $k+1$ of $X_{k+2}$, $aa^k$, are accounted for.  Therefore any additional distinct squares must end beyond the first $k+1$ positions of $X_{k+2}$. We consider two cases.
		
		 First, let us consider the case when additional distinct squares begin from positions 1 through $5k+3$ of $X_k$. (Position~$5k+3$ of $X_k$ is the second-to-last $b$ in $X_k$.) We are looking for distinct squares that end beyond the first $k+1$ positions of $X_{k+2}$, i.e., distinct squares that extend beyond $X_{k+1}$.  Since the length of $X_{k+1}$ is greater than that of $X_k$, less than half of any such square will be in $X_k$.  Therefore, all of the square that falls in $X_k$ must be part of the root of the square.  The root must therefore contain $ba^{k-1}b$, a sequence that is not found at any point beyond $X_k$.  Therefore, no $X_k$, with $k<j$, may have distinct squares in addition to those listed in Table~\ref{Tab Xk} that begin at positions 1 through $5k+3$ of $X_k$.
			
			Next, let us consider the case when additional distinct squares begin from positions $5k+4$ through $6k+3$ of $X_k$. The positions indicated spell the last $a^{k-1}b$ factor of $X_k$. Each of the last $k$ positions of $X_k$ already begins one distinct square of length $4k+6$.  For each of these $k$ positions, let $u$ be the already-established distinct square beginning at that position, and let $U$ be a theoretically longer square beginning at the same position.  We do not need to address potentially shorter distinct squares, since their non-existence is easily verified.
From Lemma~\ref{LemIlie3}, we have $\length{U}+m \leq 2\length{u}$, which gives $\length{U} \leq 8k+12-m$.  Recall that $m$ gives the number of consecutive FS-double-squares of the lengths $\length{u}$ and $\length{U}$.  The minimum value of $m$ is therefore 1, giving $\length{U} \leq 8k+11$.  It can be verified that none of the last $k$ positions of $X_k$ begin a square of length at most $2(8k+11)$.   Since no longer distinct squares may exist at those positions, no additional distinct squares exist.
		\end{proof}

		\begin{theorem}
		\label{Thm density 1}
		There exist words in which the distinct-square-density approaches $\frac{5}{6}$.
		\end{theorem}
		\begin{proof}
		Such a word can be constructed as follows. As discussed in Lemma~\ref{Lem Xk}, let $$X_k=a^{k-1}baa^{k-1}ba^{k-1}baa^{k-1}baa^{k-1}ba^{k-1}b$$ of length $6k+3$ and for $i<j$, let 
$Y_{i,j}=X_i X_{i+1}\cdots X_{j-1} X_jaa^{j-1}.$ 
Essentially, $Y_{i,j}$ is the concatenation of $j-i+1$ words $w_m$, with increasing $m$ values, in which the suffix and prefix of $a$'s are shared by adjacent pairs of words. Referring to Theorem~\ref{Thm unique}, the factor $X_jaa^{j-1}$ of $Y_{i,j}$ is the word $w_j$ of length $7j+3$, thus the word $Y_{i,j}$ has length $7j+3+\sum_{k=i}^{j-1}(6k+3)=7j+3+3j^2-3i^2$.
				
			By Theorem~\ref{Thm run}, the factor $X_jaa^{j-1}$ has $4j+\floor{\frac{j}{2}}+1$ distinct squares. By Lemma~\ref{Lem Xk}, 
			\begin{center}
			$\begin{array}{rcl}
			s(Y_{i,j})&=&s(Y_{i,j})[1..6i+3]s(Y_{i,j})[6i+4..12i+12] \cdots\\
			&&\hspace*{1in} s(Y_{i,j})[3j^2-3i^2-6j-2..3j^2-3i^2]s(X_jaa^{j-1})\\
			&=& 1^{2i+1}0^{i+1}1^i01^{2i} 1^{2(i+1)+1}0^{(i+1)+1}1^{i+1}01^{2(i+1)} \cdots\\
			&&\hspace*{1in} 1^{2(j-1)+1}0^{(j-1)+1}1^{j-1}01^{2(j-1)} s(w_j)\\
			&=&\prod_{k=i}^{j-1}(1^{2k+1}0^{k+1}1^k01^{2k})s(w_j).
			\end{array}$
				\end{center}
							
			 The word $Y_{i,j}$ has a total of
			\begin{center}
			$4j+\floor{\frac{j}{2}}+1+\sum_{k=i}^{j-1}(5k+1)=4j+\floor{\frac{j}{2}}+1+\frac{1}{2}(5j^2-3j-5i^2+3i)$
			\end{center}
			distinct squares. It has thus distinct-square-density $$\frac{4j+\floor{\frac{j}{2}}+1+\frac{1}{2}(5j^2-3j-5i^2+3i)}{7j+3+3j^2-3i^2},$$ which, recalling that $j>i$, approaches $\frac{5}{6}$ as $j$ approaches infinity.
		\end{proof}

	\begin{table}
		\begin{center}
			
			\begin{tabular}{l|l|l|l|l}
				$i$ & $j$ & distinct squares & length & distinct-square-density \\ 
				\hline
				\hline
				1  & 2  & 16   & 26    & .615 \\
				1  & 3  & 31   & 48    & .646 \\
				2  & 4  & 46   & 67    & .687 \\
				2  & 5  & 71   & 101   & .703 \\
				5  & 15  & 553  & 708   & .781 \\
				6  & 19  & 879  & 1111  & .791 \\
				8  & 25  & 1490 & 1861  & .801 \\
				11 & 36 & 3063 & 3780  & .810 \\
				19 & 64 & 9559 & 11656 & .820 \\
				\end{tabular}\caption{Densities of selected $Y_{i,j}$ words: note that both the $i$ and $j$ columns skip values.  For each $j$ value, the given $i$ value gives the maximum distinct-square-density for that $j$.}\label{Tab Y density alt}
		\end{center}
		
	\end{table}

		Referring to Table~\ref{Tab Y density alt}, we include an example of an $Y_{i,j}$ word illustrating Theorem~\ref{Thm density 1}, where $i=5$ and $j=15$; there are 553 distinct squares, the length is 708, the distinct-square-density is $\approx .781$:
\begin{center}
$\begin{array}{cccccccccccccccccccccccccc}
a&a&a&a&b&a&a&a&a&a&b&a&a&a&a&b&a&a&a&a&a&b&a&a&a&a\\
1&1&1&1&1&1&1&1&1&1&1&0&0&0&0&0&0&1&1&1&1&1&0&1&1&1\\ 
&&&&&&&&&&&&&&&&&&&&&&&&&\\
a&b&a&a&a&a&b&a&a&a&a&a&b&a&a&a&a&a&a&b&a&a&a&a&a&b\\
1&1&1&1&1&1&1&1&1&1&1&1&1&1&1&1&1&1&1&1&0&0&0&0&0&0\\ 
\end{array}
$
\end{center}

\begin{center}
$\begin{array}{cccccccccccccccccccccccccc}
a&a&a&a&a&a&b&a&a&a&a&a&a&b&a&a&a&a&a&b&a&a&a&a&a&a\\
0&1&1&1&1&1&1&0&1&1&1&1&1&1&1&1&1&1&1&1&1&1&1&1&1&1\\ 
&&&&&&&&&&&&&&&&&&&&&&&&&\\
b&a&a&a&a&a&a&a&b&a&a&a&a&a&a&b&a&a&a&a&a&a&a&b&a&a\\
1&1&1&1&1&1&1&1&1&0&0&0&0&0&0&0&0&1&1&1&1&1&1&1&0&1\\ 
&&&&&&&&&&&&&&&&&&&&&&&&&\\
a&a&a&a&a&b&a&a&a&a&a&a&b&a&a&a&a&a&a&a&b&a&a&a&a&a\\
1&1&1&1&1&1&1&1&1&1&1&1&1&1&1&1&1&1&1&1&1&1&1&1&1&1\\ 
&&&&&&&&&&&&&&&&&&&&&&&&&\\
a&a&a&b&a&a&a&a&a&a&a&b&a&a&a&a&a&a&a&a&b&a&a&a&a&a\\
1&1&1&1&0&0&0&0&0&0&0&0&0&1&1&1&1&1&1&1&1&0&1&1&1&1\\ 
&&&&&&&&&&&&&&&&&&&&&&&&&\\
a&a&a&b&a&a&a&a&a&a&a&b&a&a&a&a&a&a&a&a&b&a&a&a&a&a\\
1&1&1&1&1&1&1&1&1&1&1&1&1&1&1&1&1&1&1&1&1&1&1&1&1&1\\ 
&&&&&&&&&&&&&&&&&&&&&&&&&\\
a&a&a&a&b&a&a&a&a&a&a&a&a&b&a&a&a&a&a&a&a&a&a&b&a&a\\
1&1&1&1&1&0&0&0&0&0&0&0&0&0&0&1&1&1&1&1&1&1&1&1&0&1\\ 
&&&&&&&&&&&&&&&&&&&&&&&&&\\
a&a&a&a&a&a&a&b&a&a&a&a&a&a&a&a&b&a&a&a&a&a&a&a&a&\\
1&1&1&1&1&1&1&1&1&1&1&1&1&1&1&1&1&1&1&1&1&1&1&1&1&\\ 
&&&&&&&&&&&&&&&&&&&&&&&&&\\
a&b&a&a&a&a&a&a&a&a&a&a&b&a&a&a&a&a&a&a&a&a&b&a&a&\\
1&1&1&1&1&1&1&1&1&1&1&1&1&0&0&0&0&0&0&0&0&0&0&0&1&\\ 
&&&&&&&&&&&&&&&&&&&&&&&&&\\
a&a&a&a&a&a&a&a&b&a&a&a&a&a&a&a&a&a&a&b&a&a&a&a&a&\\
1&1&1&1&1&1&1&1&1&0&1&1&1&1&1&1&1&1&1&1&1&1&1&1&1&\\ 
&&&&&&&&&&&&&&&&&&&&&&&&&\\
a&a&a&a&b&a&a&a&a&a&a&a&a&a&a&b&a&a&a&a&a&a&a&a&a&\\
1&1&1&1&1&1&1&1&1&1&1&1&1&1&1&1&1&1&1&1&1&1&1&1&1&\\ 
&&&&&&&&&&&&&&&&&&&&&&&&&\\
a&a&b&a&a&a&a&a&a&a&a&a&a&b&a&a&a&a&a&a&a&a&a&a&a&\\
1&1&1&0&0&0&0&0&0&0&0&0&0&0&0&1&1&1&1&1&1&1&1&1&1&\\ 
&&&&&&&&&&&&&&&&&&&&&&&&&\\
\end{array}$
\end{center}

\begin{center}
$\begin{array}{cccccccccccccccccccccccccc}
b&a&a&a&a&a&a&a&a&a&a&a&b&a&a&a&a&a&a&a&a&a&a&b&a&\\
1&0&1&1&1&1&1&1&1&1&1&1&1&1&1&1&1&1&1&1&1&1&1&1&1&\\ 
&&&&&&&&&&&&&&&&&&&&&&&&&\\
a&a&a&a&a&a&a&a&a&a&b&a&a&a&a&a&a&a&a&a&a&a&a&b&a&\\
1&1&1&1&1&1&1&1&1&1&1&1&1&1&1&1&1&1&1&1&1&1&1&1&0&\\ 
&&&&&&&&&&&&&&&&&&&&&&&&&\\
a&a&a&a&a&a&a&a&a&a&b&a&a&a&a&a&a&a&a&a&a&a&a&b&a&\\
0&0&0&0&0&0&0&0&0&0&0&0&1&1&1&1&1&1&1&1&1&1&1&1&0&\\ 
&&&&&&&&&&&&&&&&&&&&&&&&&\\
a&a&a&a&a&a&a&a&a&a&a&b&a&a&a&a&a&a&a&a&a&a&a&b&a&\\
1&1&1&1&1&1&1&1&1&1&1&1&1&1&1&1&1&1&1&1&1&1&1&1&1&\\ 
&&&&&&&&&&&&&&&&&&&&&&&&&\\
a&a&a&a&a&a&a&a&a&a&a&b&a&a&a&a&a&a&a&a&a&a&a&a&a&\\
1&1&1&1&1&1&1&1&1&1&1&1&1&1&1&1&1&1&1&1&1&1&1&1&1&\\ 
&&&&&&&&&&&&&&&&&&&&&&&&&\\
b&a&a&a&a&a&a&a&a&a&a&a&a&b&a&a&a&a&a&a&a&a&a&a&a&\\
1&0&0&0&0&0&0&0&0&0&0&0&0&0&0&1&1&1&1&1&1&1&1&1&1&\\ 
&&&&&&&&&&&&&&&&&&&&&&&&&\\
a&a&b&a&a&a&a&a&a&a&a&a&a&a&a&a&b&a&a&a&a&a&a&a&a&\\
1&1&1&0&1&1&1&1&1&1&1&1&1&1&1&1&1&1&1&1&1&1&1&1&1&\\ 
&&&&&&&&&&&&&&&&&&&&&&&&&\\
a&a&a&a&b&a&a&a&a&a&a&a&a&a&a&a&a&a&b&a&a&a&a&a&a&\\
1&1&1&1&1&1&1&1&1&1&1&1&1&1&1&1&1&1&1&1&1&1&1&1&1&\\ 
&&&&&&&&&&&&&&&&&&&&&&&&&\\
a&a&a&a&a&a&a&a&b&a&a&a&a&a&a&a&a&a&a&a&a&a&b&a&a&\\
1&1&1&1&1&1&1&1&1&0&0&0&0&0&0&0&0&0&0&0&0&0&0&0&1&\\ 
&&&&&&&&&&&&&&&&&&&&&&&&&\\
a&a&a&a&a&a&a&a&a&a&a&a&b&a&a&a&a&a&a&a&a&a&a&a&a&\\
1&1&1&1&1&1&1&1&1&1&1&1&1&0&1&1&1&1&1&1&1&1&1&1&1&\\ 
&&&&&&&&&&&&&&&&&&&&&&&&&\\
a&a&b&a&a&a&a&a&a&a&a&a&a&a&a&a&b&a&a&a&a&a&a&a&a&\\
1&1&1&1&1&1&1&1&1&1&1&1&1&1&1&1&1&2&2&2&2&2&2&2&2&\\ 
&&&&&&&&&&&&&&&&&&&&&&&&&\\
a&a&a&a&a&a&b&a&a&a&a&a&a&a&a&a&a&a&a&a&a&a&b&a&a&\\
2&2&2&2&2&2&2&0&0&0&0&0&0&0&0&0&0&0&0&0&0&0&0&0&0&\\ 
&&&&&&&&&&&&&&&&&&&&&&&&&\\
a&a&a&a&a&a&a&a&a&a&a&a&b&a&a&a&a&a&a&a&a&a&a&a&a&\\
0&0&0&0&0&0&0&0&0&0&0&0&1&1&1&1&1&1&1&1&1&1&1&1&1&\\ 
&&&&&&&&&&&&&&&&&&&&&&&&&\\
a&a&a&b&a&a&a&a&a&a&a&a&a&a&a&a&a&a&a&b&a&a&a&a&a&\\
1&1&1&0&0&1&1&1&1&1&1&1&1&1&1&1&1&1&1&1&0&0&0&0&0&\\ 
&&&&&&&&&&&&&&&&&&&&&&&&&\\
a&a&a&a&a&a&a&a&a&b&a&a&a&a&a&a&a&a&a&a&a&a&a&a&a&\\
0&0&0&0&0&0&0&0&0&0&0&1&0&1&0&1&0&1&0&1&0&1&0&1&0& 
\end{array}$
\end{center}

Note that the above word does not have many FS-double-squares, and those it does have are not at the beginning.  Words with distinct-square-density greater than $.8$ first occur when $j=25$ and are well over 1,000 letters long.

\section{Selfish 2's, or not}

In all the words we have given thus far, each run of 2's in the corresponding distinct-square-sequence is followed by a run of at least twice as many 0's.  We refer to a run of 2's followed by at least twice as many 0's as {\em selfish 2's}.

However, not all runs of 2's are selfish. The most straightforward way to break the selfish rule is to omit or alter the last letter of the word $w_m$, so that the position that would be the last 2 is instead a 1.  For example consider the word $w_2=abaababaabaababaa$, which has the distinct-square-sequence $22000011100110010$.  The Selfish 2's rule appears to hold, but it can be broken simply by omitting the last letter, giving $abaababaabaababa$, which has the distinct-square-sequence $210000111011100$, or by changing the last letter of $w_2$ to $b$, giving the distinct-square-sequence $21000011101011000$.

Similar results are seen with $w_3$.  Omitting the last letter gives sequence $22100000011110011100010$ and distinct-square-density $\frac{13}{23} \approx .565$, and altering the last letter gives the distinct-square-sequence $221000000111100011100100$ and distinct-square-density $\frac{13}{24} \approx .542$.

For the above alterations to $w_2$ and $w_3$, the Selfish 2's rule very nearly holds; we have replaced a 2 with a 1, but the length of the first run of 0's remains unchanged. Greater breaks from the Selfish 2's pattern can be obtained by increasing the values of $e_1, e_2$, or both, as in the following examples.

	The distinct-square-sequence of 
	$(aba, ab, 2, 1)a$ is $22011000100011100110010$, with distinct-square-density $\approx .565$; 
	for $(aba, ab, 3, 1)a$, it is $22011011000011100011100110010$, with distinct-square-density $ \approx .586$; 
	for $(aba, ab, 3, 2)a$, it is $22011000000100011100001110111100010$, with distinct-square-density $=.514$; 
	for $(aaba, aab, 2, 1)aa$, it is $22201110000110000111100111000010$, with distinct-square-density $ \approx .594$. 
	
	The distinct square-sequence of $(aaba, aab, 2, 2)$ is $21100010000001111000011120011110001000$, with distinct-square-density $=.5$. Note that this word contains the word $(aba, a, 1, 1)$, which  does follow the Selfish 2's rule. The word $(aaba, aab, 2, 2)$ has an internal 2 which disappears when $aa$ is added, i.e., 	   
 $(aaba, aab, 2, 2)aa$ has distinct-square-sequence $2220000000000111100000111101111110000010$, with distinct-square-density $.525$  (the Selfish 2's rule applies here).  
	
	The distinct-square-sequence of $(aaba, aab, 3, 1)$ is
	$21101110111000100111100001111001101000$, with distinct-square-density $ \approx .579$. Adding one $a$ to this word gives 
	a distinct-square-sequence of $221011101110000001111000011110011100010$, with distinct-square-density $ \approx .590$. Adding another $a$ gives $2220111011100000011110000111100111000010$, with distinct-square-density $= .6 $. Note that in this case, the Selfish 2's rule does not apply even when the sequence begins with multiple 2's.
	
	The distinct-square-sequence of $(aaba, aab, 3, 2)$ is
	$211011100010000110000111100001112001111$ $0001000,$ with distinct-square-density $ \approx .523$. 
	Note that this word again contains the word $(aba, a, 1, 1)$, which does follow the Selfish 2's rule.
	Adding another $a$ to the end of the above word produces another leading 2, but causes the interior 2 to vanish, giving  
	$221011100000000110000111100001111101111100$ $00010$ with distinct-square-density $\approx .532$. The sequence can be extended to begin with three 2's by adding yet another $a$ giving $222011100000000110000111100000111101111110000010$, with distinct-square-density $ \approx  .542$.

Apart from $(aaba, aab, 2, 2)aa$, in the distinct-square-sequences of all the above examples, each run of 2's is still followed by a larger run of 0's; the difference is that the 0's are not necessarily adjacent to the 2's, and the leading 0 in the run of 0's is not necessarily the first 0 to follow the 2's.

\section{Conclusion and future work}		

In this paper, we first studied how to maximize runs of FS-double-squares in the prefix. We showed that a result of Ilie~\cite{Ili07}, which  considers the case when the lengths of squares in a run of 2's are preserved, implies that for a given positive integer $m$, the minimum length of a word beginning with $m$ FS-double-squares, whose lengths are preserved, is $7m+3$. In Theorem~\ref{Thm unique}, we constructed a word $w_m$ of length $7m+3$ that begins with $m$ FS-double-squares, whose lengths are preserved, and analyzed in Theorem~\ref{Thm run} the distinct-square-sequence as well as the distinct-square-density of $w_m$. We then generalized our construction in Theorems~\ref{Thm larger alphabet} and \ref{Thm Z words}. In Theorem~\ref{Thm density 1}, we constructed for each pair of integers $(i, j)$ with $i<j$, a word $Y_{i,j}$ in which the distinct-square-density approaches $5/6$ as $j$ approaches infinity. 

	 Deza et al. \cite{DeFrTh} gives $\lfloor \frac{5n}{6} \rfloor$ as the maximum number of FS-double-squares in a word of length $n$, and we may wonder about a connection between that result and our Theorem \ref{Thm density 1}.  	Deza et al.'s result gives an upper bound on the number of FS-double-square positions in a word; we give a pattern for a word that will have close to $\lfloor \frac{5n}{6} \rfloor$ total squares, counting both double and single-square-positions.  In fact, of all the distinct squares in our word $Y_{i,j}$, only a trivial number are FS-double-squares.  Deza et al.'s proof, on the other hand, is concerned entirely with FS-double-squares and says nothing about single distinct square occurrences.  While there may be some underlying property that leads to the value $\frac{5}{6}$ occurring in both results, neither our proof nor Deza et al.'s incorporates part of the other.
	
 We proved that the upper bound for the number of distinct squares in a word of length $n$ is at least a value approaching $\floor{\frac{5n}{6}}.$  We did so by finding a pattern for a word that when $n$ is sufficiently large, will have a distinct-square-density approaching $\frac{5}{6}$. We suspect $\lfloor \frac{5n}{6} \rfloor$ either is the upper bound or is very close to it.  However, the pattern we found approaches the distinct-square-density $\frac{5}{6}$ only for words that are thousands of letters long or more; our intuition is that there exist shorter words which approach the $\frac{5}{6}$ bound, and that finding them could be a fruitful area for future research.
	
	We also observed that many words have selfish 2's, where a run of FS-double-square-positions is followed by a longer run of positions with no distinct squares.  We disproved our first Selfish 2's hypothesis--that any run of 2's must be followed by a run of at least twice that many 0's--but we suspect that a weaker version of our Selfish 2's hypothesis is true. Proving any Selfish 2's hypothesis would put a maximum on the upper bound of distinct-square-density. With these observations in mind, we propose a weaker version of the Selfish 2's rule:  For every 2 in the distinct-square-sequence of a word, at least one 0 must exist to the right of that 2. If this rule is true, then the upper limit on the number of distinct squares in a word of length $n$ must be less than $n$ or the distinct-square-density can never be more than 1.  We suspect this is true in part because our $Y_{i,j}$ words get the vast majority of their distinct square occurrences from single rather than double-squares. Stronger versions of the rule, that require more than one 0 to follow each 2, would lead to correspondingly lower upper limits. 
	
Referring to Table~\ref{Tab Y density alt}, we suggest the problem of finding and proving the $i$ value that gives the maximum distinct-square-density for any given $j$.

The distinct-square-sequences were calculated using a program that we wrote in Java to support this paper that, given a word, outputs the associated distinct-square-sequence, the total number of distinct squares in the word, the length of the word, and the distinct-square-density of the word.  In addition to accepting typed words as input, the program also creates $Y_{i,j}$ words given $i$ and $j$ values, or creates words of the form $(v_1^{e_1}v_2v_2^{e_2})^2$ when given values for $v_1, v_2, e_1$, and $e_2$.  The words created according to those criteria then have distinct-square-sequences calculated in an identical manner to typed-in words. Distinct-square-density was calculated in Java as a 64-bit signed floating point value (Java's double type). Densities were rounded to three decimal places for convenience. Distinct-square-density values in Table~\ref{Tab Y density alt} were calculated in Microsoft Excel using the formulas given in Theorem~\ref{Thm density 1}.

\nocite{*}
\bibliographystyle{eptcs}

\providecommand{\urlalt}[2]{\href{#1}{#2}} 
\providecommand{\doi}[1]{doi:\urlalt{http://dx.doi.org/#1}{#1}}

\end{document}